\documentclass[sigconf]{acmart}
\renewcommand\footnotetextcopyrightpermission[1]{} 

\AtBeginDocument{%
  \providecommand\BibTeX{{%
    \normalfont B\kern-0.5em{\scshape i\kern-0.25em b}\kern-0.8em\TeX}}}

\setcopyright{acmcopyright}
\copyrightyear{2022}
\acmYear{2022}
\acmDOI{10.1145/1122445.1122456}

\acmConference[CIKM '22]{CIKM '22: Conference on Information and Knowledge Management}{October 17-22, 2022}{Atlanta, Georgia, USA}
\acmBooktitle{CIKM '22: Conference on Information and Knowledge Management, October 17-22, 2022, Atlanta, Georgia, USA}
\acmPrice{15.00}
\acmISBN{978-1-4503-XXXX-X/18/06}

\usepackage{citehack}
\usepackage{url}

\usepackage{amssymb}
\usepackage{bm}
\usepackage{nicefrac}
\usepackage{booktabs}
\usepackage{array}
\usepackage{multirow}
\usepackage{threeparttable}
\usepackage{makecell}
\usepackage{color}
\usepackage[procnumbered,ruled,vlined,linesnumbered]{algorithm2e}
\usepackage{enumitem}
\usepackage{stfloats}

\def\mo#1{\| #1 \|}\def\kh#1{\left( #1 \right)}

\def\mo#1{\| #1 \|}
\def\len#1{\left| #1 \right|}

\def\calG{\mathcal{G}}

\newcommand{\rea}{\mathbb{R}}
\newtheorem{property}{Property}

\newcommand\yy{\boldsymbol{\mathit{y}}}
\renewcommand\vv{\boldsymbol{\mathit{v}}}

\renewcommand\ss{\boldsymbol{\mathit{s}}}
\newcommand\zz{\boldsymbol{\mathit{z}}}
\newcommand\xx{\boldsymbol{\mathit{x}}}
\newcommand\bb{\boldsymbol{\mathit{b}}}
\newcommand\ee{\boldsymbol{\mathit{e}}}
\newcommand\ff{\boldsymbol{\mathit{f}}}
\newcommand\qq{\boldsymbol{\mathit{q}}}

\newcommand\uu{\boldsymbol{\mathit{u}}}
\newcommand\vs{\boldsymbol{\mathit{s\emph{}}}}
\newcommand\xxi{\boldsymbol{\mathit{\xi}}}
\newcommand{\eps}{\epsilon}

\newcommand\OO{\boldsymbol{\mathit{O}}}
\newcommand\LL{\bm{\mathit{L}}}
\newcommand\AAA{\boldsymbol{\mathit{A}}}
\newcommand\BB{\boldsymbol{\mathit{B}}}

\newcommand\HH{\boldsymbol{\mathit{H}}}

\newcommand\KK{\boldsymbol{\mathit{K}}}
\newcommand\DD{\boldsymbol{\mathit{D}}}
\newcommand\EE{\boldsymbol{\mathit{E}}}

\newcommand\MM{\boldsymbol{\mathit{M}}}
\newcommand\QQ{\boldsymbol{\mathit{Q}}}
\newcommand\TT{\boldsymbol{\mathit{T}}}

\newcommand\II{\boldsymbol{\mathit{I}}}
\newcommand\WW{\boldsymbol{\mathit{W}}}

\newcommand\XX{\boldsymbol{\mathit{X}}}
\newcommand\YY{\boldsymbol{\mathit{Y}}}

\newtheorem{theorem}{Theorem}[section]
\newtheorem{proposition}[theorem]{Proposition}
\newtheorem{lemma}[theorem]{Lemma}

\newtheorem{definition}[theorem]{Definition}

\DontPrintSemicolon
\SetKw{KwAnd}{and}
\SetProcNameSty{textsc}
\SetFuncSty{textsc}
\SetKwInOut{Input}{Input\ \ \ \ }
\SetKwInOut{Output}{Output}

\def\kh#1{\left( #1 \right)}
\begin{document}
\title{Effects of Stubbornness on Opinion Dynamics}

\author{Wanyue Xu}
\affiliation{%
  \institution{Fudan University}
  \city{Shanghai}
  \country{China}}
\email{xuwy@fudan.edu.cn}
\author{Liwang Zhu}
\affiliation{%
  \institution{Fudan University}
  \city{Shanghai}
  \country{China}}
\email{19210240147@fudan.edu.cn}
\author{Jiale Guan}
\affiliation{%
  \institution{Fudan University}
  \city{Shanghai}
  \country{China}}
\email{16307130212@fudan.edu.cn}
\author{Zuobai zhang}
\affiliation{%
  \institution{Fudan University}
  \city{Shanghai}
  \country{China}}
\email{17300240035@fudan.edu.cn}
\author{Zhongzhi Zhang}
\authornote{Corresponding author}
\affiliation{%
  \institution{Fudan University}
  \city{Shanghai}
  \country{China}}
\email{zhangzz@fudan.edu.cn}

\renewcommand{\shortauthors}{Xu, et al.}

\begin{abstract}
As an important factor governing opinion dynamics, stubbornness strongly affects various aspects of opinion formation. However, a systematically theoretical study about the influences of heterogeneous stubbornness on opinion dynamics is still lacking. In this paper, we study a popular opinion model in the presence of inhomogeneous stubbornness. We show analytically that heterogeneous stubbornness has a great impact on  convergence time, expressed opinion of every node, and the overall expressed opinion. We provide an explanation of the expressed opinion in terms of stubbornness-dependent spanning diverging forests. We propose quantitative indicators to quantify some social concepts, including conflict, disagreement, and polarization by incorporating heterogeneous stubbornness, and develop a nearly linear time algorithm to approximate these quantities, which has a proved theoretical guarantee for the error of each quantity. To demonstrate the performance of our algorithm, we perform extensive experiments on a large set of real networks, which indicate that our algorithm is both efficient and effective, scalable to large networks with millions of nodes.
\end{abstract}


\maketitle

\section{Introduction}
As a research subject of computational social science~\cite{HoWaAtetal21}, opinion dynamics has received considerable attention from the scientific community~\cite{AnDaPrRaYe20}, especially in recent years due to the explosive growth of social media and online social networks~\cite{Le20}. An important factor affecting opinion dynamics is the attributes of individuals, for example, stubbornness that measures the extent to which an individual sticks to his own initial opinion/belief or is willing to change his opinion in reaction to its neighbors' opinions. Thus, stubbornness indicates people's susceptibility to persuasion~\cite{AbKlPaTs18, AbChKlLiPaSoTs21}. In social psychology, a rich body of empirical work~\cite{AlPoMc06,DiKlLe92,EvOaSc92,StMo07, WaPoAr93} focused on this key element, unveiling some far-reaching implications of stubbornness in diverse aspects, including product marketing, public health campaigns, and political candidates. Some recent work studied persuasive technologies to change an individual's attitude or stubbornness~\cite{Fo02,IjDeMi06,KaMade09}, showing that individual's stubbornness are heterogeneous.

Perhaps the most important step for opinion dynamics study is to establish a mathematical model by capturing various fundamental factors affecting opinion formation. During the past decades, a large body of literature has been devoted to the modeling of opinion dynamics~\cite{PrTe17}, and several classic models have been developed, providing a deep understanding of opinion diffusion and formulation. Among these models, Friedkin-Johnsen (FJ) model~\cite{FrJo90} is one of the most popular models. As a significant extension of the Degroot model~\cite{De74}, the FJ model grasps some intricate social behaviors by considering French's ``theory of social power''~\cite{Fr56}. Very recently, by modifying the FJ model a study was initiated in~\cite{AbKlPaTs18, AbChKlLiPaSoTs21} to explore how to change individuals' stubbornness to optimize the overall opinion. However, there is still no theoretical studies that systematically uncover the influences of heterogeneous stubbornness on opinion dynamics, including convergence speed, overall opinion, and some resulting social phenomena such as polarization, disagreement, and conflict.  

In order to fill the aforementioned gap, in this paper, we study the FJ model for opinion dynamics on an undirected weighted graph, aiming at unveiling the role played by heterogeneous stubbornness in opinion dynamics. The main contributions of this paper are summarized below. 
\begin{itemize}[leftmargin=*]
\item We give an explanation of the expressed opinion in terms of spanning diverging forests of a corresponding defined directed graph. 

\item We show that heterogeneous stubbornness has a strong effect on various aspects of opinion dynamics, including convergence velocity, expressed opinion of every node, and the overall expressed opinion. 

\item We make an extension of some relevant quantities for the FJ model with  uniform stubbornness to the case with heterogeneous stubbornness, including conflict, disagreement, and polarization by incorporating heterogeneous stubbornness. We then develop an approximation algorithm to evaluate these quantities, which has a nearly linear time complexity and a theoretically guaranteed error. 

\item We perform extensive experiments on large datasets, which demonstrate that our algorithm achieves both good efficiency and effectiveness, being scalable to large networks with millions of nodes.
\end{itemize}

\section{Preliminaries}\label{preliminaries}

In this section, we briefly introduce some concepts about directed and undirected  graphs,  their related matrices, as well as some useful notations.

\subsection{Digraph and Related Concepts}

Let $\Gamma=(V(\Gamma),E(\Gamma), w)$ denote a weighted directed graph (network) abbreviated as digraph with  node set $V(\Gamma)$, arc (edge) set $E(\Gamma)$ and edge weight function $w : E(\Gamma) \to \rea_{+}$. Suppose $n=|V(\Gamma)|$ and $V(\Gamma)=\{v_1,v_2,...,v_n\}$. In what follows, $v_i$ and $i$ are used interchangeably to denote node $v_i$ in the case without inducing any confusion. Let $i\to j$ denote an arc in $ E(\Gamma)$ pointing to $j$ from $i$. In $\Gamma$, a loop is an arc having identical end nodes, and an isolated node is a node having no arcs pointing to or coming from it. For two nodes $v_1$ and $v_k$ in digraph $\Gamma$, a path from $v_1$ to $v_k$  is an alternating sequence of nodes and arcs $v_1, e_1, v_2,..., e_{k-1}, v_k$, in which nodes are distinct and every arc $e_i$ is $v_{i}\to v_{i+1}$. A circuit is a path plus an arc from the ending node to the starting node. Digraph $\Gamma$ is called (strongly) connected if there exists at least a path for any pair of nodes. $\Gamma$ is called weakly connected if it is connected when one replaces any arc $i\to j$ by two arcs $i\to j$ and $j\to i$ in opposite directions. A directed tree is a weakly connected digraph with no loops and cycles. An isolated node is considered as a tree. A directed forest is a particular digraph that is a disjoint union of directed trees.

For  a weighted digraph $\Gamma$, many connection properties are encoded in its weighted adjacency matrix $\AAA=(w_{ij})_{n \times n}$, whose entry $w_{ij}$ at row $i$ and column $j$ represents the weight of arc $i\to j $. For each arc $i\to j \in E$ in $\Gamma$, its weight $w_{i j}$ is strictly nonnegative, representing the influence of $j$ on $i$. If there is no arcs from $i$ to $j$, $w_{i j}=0$. For a node $i$, its weighted in-degree $d^+_i$ is defined by  $d^+_i=\sum_{j=1}^n w_{ji}$, and its weighted out-degree $d^-_i$ is $d^-_i=\sum_{j=1}^n w_{ij}$. Let $\Theta_i$ be the set of nodes to which node $i$ points. Then, $d^-_i$ is also equal to $\sum_{j\in \Theta_i } w_{ij}$. In the sequel, we use $d_i$ to represent the out-degree $d_i^-$. The diagonal degree matrix of $\Gamma$ is $\DD= {\rm diag}(d_1, d_2, \ldots, d_n)$,  with the $i$-th diagonal entry being $d_i$. The Laplacian matrix of $\Gamma$ is $\LL=(l_{ij})_{n \times n}=\DD-\AAA$. Let $\mathbf{1}$ and $\mathbf{0}$ be the two  $n$-dimensional vectors with all entries being ones and zeros, respectively.  Let $\OO$ represent the zero matrix of appropriate dimension. By definition, the sum of all entries in each row of $\LL$ equals $0$, implying $\LL\mathbf{1}=\mathbf{0}$. Let $\II$ be the identity matrix. Then matrix $\II+\LL$ is invertible, with each entry of matrix $\left(\II+\LL\right)^{-1}$ being positive~\cite{ChSh97,ChSh98}.

\subsection{Undirected Graph and Its Laplacian Matrix}

In a weighted digraph $\Gamma$, if for any arc $i\to j $ with weight $w_{ij}$, the arc $ j \to i $ exists and has weight $w_{ji}$ equal to $w_{ij}$, then $\Gamma$ is called an undirected weighted graph, denoted by ${\calG}=  (V({\calG}), E({\calG}), w)$. Arcs $ j \to i $ and  $ i \to j $ together form an undirected edge in $\calG$. Let $w_{\rm max}$ and $w_{\rm min}$ denote, respectively, the maximum and minimum weight among all edges in $E(\calG)$ of graph $\calG$. By definition, for an undirected graph $\calG$,  $w_{ij}= w_{ji}$ holds for any pair of nodes $i$ and $j$, and  $d^+_i= d^-_i$ holds for any node $i$. Moreover, both the weighted adjacency matrix $\AAA$ and Laplacian matrix $\LL$ of $\calG$ are symmetric, with the latter obeying $\LL\mathbf{1}=\mathbf{0}$.

For an undirected weighted graph ${\calG}=  (V({\calG}), E({\calG}), w)$  with $n$ nodes and $m$ edges, there is an alternative construction for its  Laplacian matrix $\LL$. Let $\BB \in \mathbb{R}^{| E({\calG})| \times |V({\calG})|}$ represent the signed edge-node incidence matrix of  $\calG$. It is an $m\times n$ matrix, whose entry $b_{ev}$ with  $e\in E(\calG)$ and $ v\in V(\calG)$ is defined as follows: $b_{ev}=1$ (resp. $b_{ev}=-1$) if node $v$ is the head (resp. tail) of edge $e$,  and $b_{ev}=0$ otherwise. Let $\ee_i$ denote the $i$-th standard basis vector. In matrix $\BB$, the row vector $\bb_{e}$ corresponding to $e$ linking two nodes $i$ and $j$ can be written as $\bb_{e}=\ee_i-\ee_j$.  We use  $\WW$ to denote  the $m\times m$ diagonal weight matrix ${\rm diag}(w_1, w_2, \ldots, w_m)$, the $e$-th diagonal entry of which is equal to the weight of edge $w_e$. Then the Laplacian matrix $\LL$ of $\calG$ can also be written as $\LL = \BB^\top \WW \BB=\sum\nolimits_{e\in E(\calG)} w_e \bb_e \bb_e^\top $, implying that $\LL$ is symmetric and positive semidefinite.

The positive semidefiniteness of Laplacian matrix $\LL$ for graph $\calG$ means that   all  eigenvalues of $\LL$  are non-negative. Moreover, when graph $\calG$ is connected, its  Laplacian matrix $\LL$ has a unique zero eigenvalue. Let $\lambda_1 \ge  \lambda_2 \ge \ldots \ge\lambda_{n-1} \ge \lambda_n=0$ be the $n$ eigenvalues of  $\LL$ for connected $\calG$. Then,  $\lambda_{\max}= \lambda_{1}\leq n\, w_{\rm max}$~\cite{SpSr11}, and $\lambda_{\min}=\lambda_{n-1}\geq w _{\min}/ n^2 $~\cite{LiSc18}. Moreover, it has been shown that for any diagonal matrix $\KK={\rm diag}(k_1, k_2, \ldots, k_n)$ with $k_i>0$ ($i=1,2, \ldots, n$), matrix $\KK +\LL$ is invertible, and every entry of matrix $\left(\KK +\LL\right)^{-1}$ is positive~\cite{LiPeShYiZh19}.

\subsection{Some Useful Notations}

In this subsection, we introduce some notations and their properties, which are useful for our description and proofs. For two nonnegative scalars $a\geq 0$ and $b\geq 0$, $a$ is called an $\eps$-approximation of $b$ and denoted by $a \approx_{\eps} b$, if $(1-\eps) a \leq b \leq (1+\eps) a$ for $0\leq \eps \leq1/2$. The $\eps$-approximation has the following basic properties. For nonnegative scalars $a$, $b$, $c$, and $d $, if $a \approx_\eps b$ and $c \approx_\eps d$, then $a + c \approx_\eps b + d$. For two matrices $\XX$ and $\YY$, if $\YY - \XX$ is positive semidefinite we write $\XX \preceq \YY$. In other words, for every real vector $\xx$, the relation $\xx^\top \XX \xx \leq \xx^\top \YY \xx$ holds. For a diagonal matrix $\KK={\rm diag}(k_1, k_2, \ldots, k_n)$, we use $k_{\rm max}$ and $k_{\rm min}$ to denote, respectively, the maximum and minimum values of its diagonal entries. Then, for the Laplacian matrix $\LL$ of a connected undirected graph $\calG$, the following expressions hold true: $\KK \preceq  \LL+\KK$,  $\LL \preceq \LL+\KK$,  $ \LL+\KK \preceq (k_{\rm max} +nw_{\rm max}) \II$, and $ \frac{1}{k_{\rm max} +n\,w_{\rm max}}\LL \preceq \II$.

\section{FJ Model with Heterogeneous Stubbornness }\label{model}


In this section, we study the FJ model~\cite{FrJo90}  of  opinion dynamics on  an undirected weighted graph ${\calG}=  (V({\calG}), E({\calG}))$ with  heterogeneous stubborn agents.  We first introduce the FJ model on a connected digraph $\Gamma$ with homogeneous stubbornness.

\subsection{ FJ  Model on a Digraph with Unit Stubbornness}

The FJ model is a classic model for opinion formulation~\cite{FrJo90}. For the FJ model on a digraph $\Gamma=(V(\Gamma),E(\Gamma), w)$, each node $i \in V$ has two different opinions,  internal (or innate) opinion $s_i$ and expressed opinion $z_i$. Moreover, every individual displays some stubbornness $k_i$ against changing its opinion. The internal opinion $s_i$  is assumed to remain  constant in the interval $[-1,1]$,   while the  expressed opinion $z_i(t+1)$ at time $t+1$ evolves as a weighted average of  its internal opinion $s_i$ and the expressed opinions of its neighbors at time $t$ as
\begin{equation}\label{FJmodel}
z_i(t+1)=\frac{ s_{i}+\sum_{j \in \Theta_i} w_{i j} z_{j}(t)}{1+\sum_{j \in \Theta_i} w_{i j}}.
\end{equation}
Note that during the above updating process of expressed opinions, the stubbornness $k_i$ of every node $i$ is 1, which is a common assumption in the literature. After evolution of enough time, the expressed opinions of nodes converge. Let $\zz=(z_1, z_2,\ldots,z_n)^\top$ be the vector of expressed opinions, where $z_i$ is the  expressed opinion of node $i$ at equilibrium. It was shown~\cite{BiKlOr15} that the unique equilibrium opinion vector $\zz$  is
\begin{equation}\label{FJmodel02}
\zz =(\II+\LL)^{-1} \vs\,.
\end{equation}

In~\cite{MaTeTs17}, $\mathbf{\Omega}=(\II+\LL)^{-1}=(\omega_{ij})_{n \times n}$ is called the fundamental matrix of the FJ opinion dynamics model. It has been proved that matrix $\mathbf{\Omega}$ is row stochastic satisfying  $\mathbf{\Omega} \textbf{1}=\textbf{1}$~\cite{ChSh97,ChSh98}. Then,~\eqref{FJmodel02} indicates that for any node $i \in V$, its expressed opinion $z_i$ at equilibrium is  $z_i=\sum^n_{j=1}  \omega_{ij}s_j$, which is a convex combination of the internal opinions of all nodes including $i$ itself. Since for $i \neq j$, $\omega_{ii}>\omega_{ij}$ holds~\cite{ChSh97,ChSh98}, $z_i$ has a bias towards the internal opinion $s_i$. That is, at equilibrium any node is inclined to keep its internal opinion of its own rather than others.

\subsection{ FJ  Model on a Graph with Heterogeneous Stubbornness}


Many research works~\cite{Fo02,IjDeMi06,KaMade09}  show that the stubbornness of  individuals are heterogeneous. Here we study FJ  model on an undirected weighted graph ${\calG}=  (V({\calG}), E({\calG}), w)$  with inhomogeneous stubbornness, with an aim to explore the influences of  inhomogeneous stubbornness on various aspects of opinion dynamics.

\subsubsection{Evolution Rule}

Different from the~\eqref{FJmodel}, for the FJ model on graph ${\calG}$,  each node $i$ has different inclination to his initial opinion. At time $t+1$, the expressed opinion of node $i$ updates in the following way:
\begin{equation}\label{SA02}
    z_i(t+1)=\frac{k_i s_i+\sum_{j\in\Theta_{i}}w_{ij}z_j(t)}{k_i+\sum_{j\in\Theta_{i}}w_{ij}},
\end{equation}
where the stubbornness coefficient $k_i\ge 0$ measures the  stubbornness of node $i$ to his initial opinion. We consider the case that the  stubbornness coefficients of all nodes are greater than 0 obeying $k_1k_2\ldots k_n\ne 0$. In what follows, we call  ${\KK} = {\rm diag}(k_1, k_2, \ldots, k_n)$ the stubbornness matrix of  $\calG$.

\subsubsection{Existence and Expression of Equilibrium Opinions}

We now determine the  expressed opinions at equilibrium. We first express~\eqref{SA02} in matrix form as 
\begin{equation}\label{SA03}
    \zz(t+1)=\QQ\AAA\zz(t)+\QQ\KK \vs,
\end{equation}
where $\QQ= {\rm diag}(q_1, q_2, \ldots, q_m)$  is an $n \times n$ diagonal matrix with the $i$-th diagonal entry $q_i$ being
$q_i={1}/({k_i+\sum_{j\in\Theta_{i}}w_{ij}})$.

By iterating~\eqref{SA03}, the expressed opinion vector $\zz(t)$ at time $t$ is obtained to be 
\begin{equation}\label{SA05}
    \zz(t)=(\QQ\AAA)^t\zz(0)+\sum_{s=0}^{t-1}(\QQ\AAA)^s\QQ\KK\zz(0).
\end{equation}
Note that  the sum of the $i$-th row for matrix $\QQ\AAA$ is $ \frac{\sum_{j\in\Theta_{i}}w_{ij}}{k_i+\sum_{j\in\Theta_{i}}w_{ij}}<1$ for any $i\in V$,  implying
that matrix $\QQ\AAA$ is row sub-stochastic. According to Perron-Forbenius theorem~\cite{Ma00}, the absolute values of all  eigenvalues of matrix $\QQ\AAA$ are less than 1. Thus, we have
$\lim_{t\rightarrow \infty}(\QQ\AAA)^t=\OO$. 
Then, we can conclude the existence of  vector $\zz$ for equilibrium expressed opinions, which reads 
\begin{align}\label{noisesteady}
\zz    =\lim_{t\rightarrow \infty}\zz(t) =(\LL+\KK)^{-1}\KK\zz(0)
     =(\LL+\KK)^{-1}\KK \vs\,.
	\end{align}

Let $\mathbf{\Phi}=(\phi_{ij})_{n \times n}\triangleq(\LL+\KK)^{-1}\KK$ be the fundamental matrix of this opinion model. Due to  $\LL\textbf{1}=\textbf{0}$, relation $(\LL+\KK)\textbf{1}=\KK\textbf{1}$ holds. On the other hand, matrix $\LL+\KK$ is invertible, then $(\LL+\KK)^{-1}\KK\textbf{1}=\textbf{1}$, implying that $\mathbf{\Phi}=(\LL+\KK)^{-1}\KK$ is row-stochastic. Equation~\eqref{noisesteady} shows that, for each node $i \in V$ its expressed opinion $z_i$ is  $z_i=\sum^n_{j=1}  \phi_{ij}s_j$, which is a convex combination of the internal opinions for all nodes in $V$. Since  $\mathbf{\Phi}$ is  determined by network structure, edge weighs, and the levels of stubbornness of nodes,  its is the same with the equilibrium expressed opinion for every node when the internal opinions are fixed.

\subsubsection{Properties of  Opinions}

Using the properties of the fundamental matrix $(\LL+\KK)^{-1}\KK$, we give some relations between internal opinion vector $\vs$ and expressed opinion vector $\zz$.
\begin{property}\label{Mean}
For a graph ${\calG}=  (V({\calG}), E({\calG}), w)$, if the weighted sum of internal opinions $\textbf{1}^{\top}\KK\vs$ is 0, then the weighted sum of expressed opinions  $\textbf{1}^{\top}\KK\zz$ in equilibrium is also 0.
\end{property}
\begin{property}\label{Offset}
For  opinion dynamics on a graph ${\calG}=  (V({\calG}), E({\calG}), w)$, if the internal opinion vector $\vs$ is changed to  $\vs'=\vs+c\textbf{1}$ with $c$ being a constant, then the corresponding equilibrium  opinion vector is changed from $\zz$ to $\zz'=\zz+c\textbf{1}$.
\end{property}
Thus, without loss of generality, in the sequel  we assume that the internal opinions obey $\textbf{1}^{\top}\KK\vs=0$, which means $\textbf{1}^{\top}\KK\zz=0$. Otherwise, if $\textbf{1}^{\top}\KK\vs \neq 0$, we change $\vs$ to
$\vs'=\vs-\frac{\textbf{1}^{\top}\KK\vs}{n}\textbf{1}$. By Property~\ref{Offset},  the expressed opinion vector $\zz$ is changed to $\zz'=\zz-\frac{\textbf{1}^{\top}\KK\ss}{n}\textbf{1}$.
\subsection{Interpretation in Terms of Spanning Diverging Forests}
Here we provide an explanation for the expressed opinions $\zz=(\LL+\KK)^{-1}\KK \vs$. For this purpose, we 
map ${\calG}=  (V({\calG}), E({\calG}), w)$ with  stubbornness matrix $\KK$  into a weighted digraph $\mathcal{G}'$.  The node set  $V'$ of digraph  $\mathcal{G}'$ is the same as the node set $V$ of graph ${\calG}$. For each edge $e$ with weight $w_{ij}=w_{ji}$  in ${\calG}$ linking two end nodes $i$ and $j$, we create two arcs $i\to j$ and $j \to i$ in $\mathcal{G}'$, the weights of which are $w_{ij}/k_i$ and  $w_{ji}/k_j$, respectively. Recall that the fundamental matrix $\mathbf{\Phi}$ for opinion dynamics on graph ${\calG}=  (V({\calG}), E({\calG}), w)$  with  stubbornness matrix $\KK$ can be written as $\mathbf{\Phi}=(\LL+\KK)^{-1}\KK=(\II+\KK^{-1}\LL)^{-1}$. Since $\KK^{-1}\LL$ is the Laplacian matrix of  digraph $\mathcal{G}'$, according to~\eqref{FJmodel02}, it is easy to verify that $\mathbf{\Phi}$ is in fact the fundamental matrix of opinion dynamics on $\mathcal{G}'$, when the stubbornness of all nodes is 1. Below, we give an interpretation of expressed opinions $\zz=(\LL+\KK)^{-1}\KK \vs$ in terms of spanning converging forests of digraph $\mathcal{G}'$.

For a digraph $\mathcal{G}'$ with node set $V'$ and arc set $E'$, a subdigraph of $\mathcal{G}'$ is a digraph, whose node set and arc set are subsets of $V'$ and $E'$, respectively, but the node set $V'$ must comprise the head and tail of each arc in set $E'$. A spanning subdigraph of $\mathcal{G}'$ is a subdigraph with node set $V'$. A spanning diverging tree or an out-tree is a weakly connected digraph in which one node, called the root, has in-degree $0$ and the remaining nodes have in-degree $1$. An isolated node is considered as a diverging tree with the root being itself. A spanning diverging forest of $\mathcal{G}'$ is a spanning subdigraph of $\mathcal{G}'$, where every weakly connected component is a diverging tree. A spanning diverging forest is also called a spanning out-forest~\cite{AgCh01,ChAg02}.

For a subdigraph $\mathcal{\bar{G}}$ of digraph $\mathcal{G}'$, its weight $\varepsilon(\mathcal{\bar{G}})$ is defined as the product of the weights of all arcs in $\mathcal{\bar{G}}$. If $\mathcal{\bar{G}}$ has no arc, its  weight $\varepsilon(\mathcal{\bar{G}})$ is set to be $1$. For a nonempty set $S$ of subdigraphs, define its weight $\varepsilon(S)$  as $\varepsilon(S) =\sum_{\mathcal{\bar{G}} \in S}  \varepsilon(\mathcal{\bar{G}})$. If $S$ is empty, $\varepsilon(S)$ is defined to be zero~\cite{ChSh97,ChSh98}. For a digraph $\mathcal{G}'$, let $\Upsilon$ be the set of all its out-forests, and $\Upsilon_{i j}$ the set of out-forests with nodes $i$ and $j$ in the same out-tree rooted at node $i$. In~\cite{ChSh97,ChSh98}, the  diverging forest matrix of a digraph $\mathcal{G}'$ is proposed and defined  as  $\mathbf{\Phi}=(\phi_{ij})_{n \times n}=\left(\II+\KK^{-1}\LL\right)^{-1}$, where the $ij$-th entry $\phi_{ij}$ is  $\phi_{ij}=\varepsilon(\Upsilon_{i j})/ \varepsilon(\Upsilon)$. Obviously, the diverging forest matrix $\mathbf{\Phi}$ is exactly the fundamental matrix of opinion dynamics on $\mathcal{G}'$, when the stubbornness matrix is an identity matrix.

In this way, we have provided an interpretation  of the equilibrium  expressed opinion vector $\zz$ given in~\eqref{noisesteady} for a undirected graph $\mathcal{G}$, in terms of  the diverging  forest matrix of  a corresponding digraph $\mathcal{G}'$, where edge weights contain the  information of node stubbornness. The perspective of our interpretation is novel, which is different from previous ones, such as absorbing random walks~\cite{GiTeTs13}, game theory~\cite{BiKlOr15}, and electrical networks~\cite{GhSr14}.

\section{Effects of Heterogeneous Stubbornness on Opinion Dynamics}\label{effects}

In this section, we  give  insight on how the stubbornness affects the convergence velocity to  equilibrium opinions, overall opinion, and the  expressed opinion of every node.
\subsection{Convergence Time}
Equation~\eqref{noisesteady}  characterizes the equilibrium expressed opinions. This  equilibrium behavior is relevant only if it  converges in a reasonable time~\cite{El93}. Below study the impact of the level of stubbornness on  the convergence time. We first give a definition of  error vector for expressed opinion vector $\zz(t)$  at time $t$. For brevity, in the sequel, we use $\calG= (V,E,w)$ to represent ${\calG}=  (V({\calG}), E({\calG}), w)$.
\begin{definition}
For opinion dynamics on a graph  $\calG= (V,E,w)$, the error vector $\ee(t)$ of expressed opinion vector  $\zz(t)$ at time $t$
is the difference of $\zz(t)$ and the equilibrium opinion vector $\zz$:
$\ee(t)=\zz(t)-\zz$.
\end{definition}

Based on~\eqref{SA05},  $\ee(t)$ can  be expressed  as
\begin{equation}
    \ee(t)=(\QQ\AAA)^t\Big(\zz(0)-\sum_{s=0}^{\infty}(\QQ\AAA)^s\QQ\KK\zz(0)\Big),\nonumber
\end{equation}
which leads to the following recursive relation governing $\ee(t)$ and $\ee(t+1)$:
\begin{equation}\label{CT01}
    \ee(t+1) = \QQ\AAA \ee(t).
\end{equation}
In order to give a better analysis of the convergence time, we introduce a new  error vector $\ff(t)$ defined as $\ff(t)=\QQ^{-\frac 1 2}\ee(t)$, the $i$-th entry of which is $f_i(t)=\frac{\ee_i(t)}{\sqrt{d_i+k_i}}$.
Using~\eqref{CT01}, we obtain
\begin{equation}\label{CT02}
    \ff(t+1) = \QQ^{\frac 1 2}\AAA\QQ^{\frac 1 2} \ff(t).
\end{equation}

We now define the convergence time of opinion dynamics on graph  $\calG= (V,E,w)$ and give an upper bound for it.
\begin{definition}
For opinion dynamics on a graph  $\calG= (V,E,w)$, its convergence time $t(\varepsilon)$ is defined as
\begin{equation}
    t(\varepsilon)=\inf \{t\ge 0, |\ff(t)|\le \varepsilon\}.
\end{equation}
\end{definition}

Let $\rho_{{\rm max}}$ be the spectral radius of matrix $\QQ\AAA$. Since $\QQ\AAA$ is  sub-stochastic, $\rho_{{\rm max}}$ is always less than 1.
\begin{proposition}\label{convergence}
For opinion dynamics on a graph  $\calG= (V,E,w)$, its convergence time $t(\varepsilon)$ is not greater than $\log_{\rho_{{\rm max}}} \varepsilon - \log_{\rho_{{\rm max}}} |\ff(0)|$.
\end{proposition}
\begin{proof}
	Since $\QQ^{\frac 1 2}\AAA\QQ^{\frac 1 2}=\QQ^{-\frac 1 2}(\QQ\AAA)\QQ^{\frac 1 2}$, $\QQ^{\frac 1 2}\AAA\QQ^{\frac 1 2}$ and $\QQ\AAA$ has the same set of eigenvalues.    Let $1>\rho_{{\rm max}}= |\rho_1| \ge |\rho_2| \ge \ldots \ge |\rho_n|$ be  the eigenvalues of the symmetric matrix $\QQ^{\frac 1 2}\AAA\QQ^{\frac 1 2}$, and let $\xxi_1$, $\xxi_2$,$\ldots$,$\xxi_n$ be their corresponding orthogonal eigenvectors. Then $\ff(t)$ can be recast  as
	\begin{equation}
	\ff(t)=\sum_{i=1}^n \frac{\xxi_i^\top \ff(t)}{\xxi_i^\top \xxi_i}\xxi_i.\label{ft}
	\end{equation}
	Based on~\eqref{CT02} and~\eqref{ft}, we obtain the recursive relation between $\ff(t+1)$ and $\ff(t)$ as
	\begin{equation}
	\ff(t+1)=\QQ^{\frac 1 2}\AAA\QQ^{\frac 1 2} \ff(t)=\sum_{i=1}^n \frac{\xxi_i^\top \ff(t)}{\xxi_i^\top \xxi_i}\rho_i\xxi_i.
	\end{equation}
	Then, we have
	\begin{align*}
	|\ff(t+1)|^2&=\sum_{i=1}^n \left(\frac{\xxi_i^\top \ff(t)}{\xxi_i^\top \xxi_i}\right)^2 \rho_i^2\xxi_i^2\le \rho^2_{{\rm max}}\sum_{i=1}^n \left(\frac{\xxi_i^\top \ff(t)}{\xxi_i^\top \xxi_i}\right)^2\xxi_i^2\\
	&= \rho^2_{{\rm max}}  |\ff(t)|^2.
	\end{align*}
	Thus, $|\ff(t+1)|\le \rho_{{\rm max}} |\ff(t)|$ and $|\ff(t)|\le \rho^t_{{\rm max}} |\ff(0)|$, which means that $|\ff(t)|$ decreases exponentially with $|\ff(t)|$. When $|\ff(t)|$ drops below $\varepsilon$ with $ 0 \le \varepsilon \le |\ff(0)| $, the convergence time is not great than $\log_{\rho_{{\rm max}}} \varepsilon - \log_{\rho_{{\rm max}}} |\ff(0)|$.
\end{proof}

Proposition~\ref{convergence} indicates that  the convergence time of opinion dynamics on  graph  $\calG= (V,E,w)$ is related to the leading eigenvalue $\rho_{{\rm max}}$ of matrix  $\QQ\AAA$, with smaller $\rho_{{\rm max}}$ corresponding to faster convergence.

The next theorem shows the change of the largest eigenvalue  $\rho_{{\rm max}}$ for matrix  $\QQ\AAA$ with respect to the stubbornness $k_i$ of node $i$.  
\begin{theorem}\label{ThmRho}
For opinion dynamics on an undirected graph  $\calG$ with stubbornness matrix $\KK$, the largest eigenvalue  $\rho_{{\rm max}}$ of matrix  $\QQ\AAA$ is a decreasing function of stubbornness $k_i$ of node $i$, if the stubbornness of all other nodes is fixed.
\end{theorem}
\begin{proof}
	Let $\uu$ be the unit eigenvector of matrix  $\QQ\AAA$ corresponding to the largest eigenvalue  $\rho_{{\rm max}}$. If we increase the stubbornness $k_i$ of  node $i$ to $k_i'$, while fixing the values of the stubbornness of other nodes, we  use   $\rho'_{{\rm max}}$ and $\vv$ to represent, respectively, the  largest eigenvalue  and its corresponding unit eigenvector for the resulting matrix $\QQ' \AAA$. By Perron-Frobenius theorem~\cite{Ma00}, all elements in $\uu$ and $\vv$ are positive. Then, we have
	\begin{equation}
	\rho_{{\rm max}}=\uu^\top \QQ\AAA \uu \geq \vv^\top \QQ\AAA \vv > \vv^\top  \QQ' \AAA \vv =\rho'_{{\rm max}},\nonumber
	\end{equation}
	which finishes the proof.
\end{proof}

Thus, increasing the stubbornness value of a node can accelerate the convergence of  opinion dynamics. On the contrary,  reducing the level of stubbornness of a node may deteriorate the convergence speed.
\subsection{Overall Expressed Opinion}
We continue to study  the influence of inhomogeneous stubbornness on the overall opinion, that is, the sum of expressed opinions of all nodes, which has received considerable recent attention~\cite{GiTeTs13,AbKlPaTs18,ZhZh21}.

It was shown in~\cite{GiTeTs13} that for the FJ model of opinion dynamics on an undirected graph with each node having a unit stubbornness, the overall opinion is equal to the total internal opinion, that is $\sum^n_{i=1}  s_i=\sum^n_{i=1}z_i $, although the equilibrium  expressed opinion for an individual node may differ from its  internal  opinion. Moreover, this conservation law does not rely on the network structure and the distribution of edge weights. This conservation law can also be understood as follows. Equation~\eqref{noisesteady} shows  when $\KK=\II$, the fundamental matrix is reduced to $\mathbf{\Phi}=(\II+\LL)^{-1}$, which is symmetric and doubly stochastic~\cite{Me97} obeying $\sum_{j=1}^{n}\phi_{ij}=1$ for $i=1,2,\ldots,n$, and $\sum_{i=1}^{n}\phi_{ij}=1$ for $j=1,2,\ldots,n$. Then, $\sum^n_{i=1}z_i =\sum^n_{i=1} \sum_{j=1}^n \phi_{ij} s_j = \sum^n_{j=1}\sum^n_{i=1} \phi_{ij} s_j  =\sum^n_{i=1}  s_i $. Actually, for any  $\KK=k \II$ with $k$ being a positive constant,  $\sum^n_{i=1}  s_i=\sum^n_{i=1}z_i $ always holds.

In the case that the stubbornness matrix $\KK$ is not a scalar matrix, the opinion dynamics on undirected graph $\mathcal{G}$ is  equivalent to opinion dynamics on a digraph $\mathcal{G}'$ with stubbornness matrix being $\II$. Note that for any pair of arcs $i\to j$ and $j \to i$ in digraph $\mathcal{G}'$, their weights are usually unequal, with the exception of $k_i= k_j$. Thus, the Laplacian matrix $\KK^{-1}\LL $ of digraph $\mathcal{G}'$ is asymmetric, it is the same with the diverging forest matrix $\mathbf{\Omega}=(\KK^{-1}\LL+\II)^{-1}$. This asymmetry is stemmed from the heterogeneous stubbornness, which is ubiquitous in social systems. Then matrix $\mathbf{\Omega}=(\KK^{-1}\LL+\II)^{-1}$ might not be column stochastic, although it is always row stochastic. This reciprocity may lead to the disappearance of conservation of total opinions, namely, $\sum_{i=1}^n z_i\neq \sum_{i=1}^n s_i$. It is not difficult to verify that heterogeneous stubbornness can not only increase but also decrease the overall opinion, exerting a substantial influence on the overall opinion.
\subsection{Equilibrium Expressed Opinions of Nodes}
As shown in~\eqref{noisesteady},  the expressed opinion $z_i$ of a node $i \in V$ is  $z_i=\sum^n_{j=1}  \phi_{ij}s_j$, where $\phi_{ij}$ is the $ij$-entry of the fundamental matrix $\mathbf{\Phi}$. Since  $\mathbf{\Phi}$  depends on the stubbornness matrix, it is expected that $\phi_{ij}$  is related to the level of stubbornness of every node. Next we study the influence of  node stubbornness on the entries of matrix $\mathbf{\Phi}$.
\begin{theorem}\label{stub}
For opinion dynamics on an undirected graph $\calG= (V,E,w)$ with Laplacian matrix $\LL$ and stubbornness matrix $\KK$, if the stubbornness $k_v$ of node $v \in V$ is decreased while the stubbornness of all other nodes in $V$ is fixed, then for matrix $\mathbf{\Phi}=(\LL+\KK)^{-1}\KK$, the elements in column $v$ decrease, while elements not in  $v$ column increase.
\end{theorem}
\begin{proof}
	Let $\EE_{v,v}$ be the $n \times n$ matrix, which has only one nonzero entry 1 at  row $v$  and column $v$, while other  entries  are zeros. For matrix $\KK= {\rm diag}(k_1,k_2, \ldots, k_n)$, if $k_v$ is decreased while other $k_i$ $(i\neq v)$ keeps  unchanged,  then $\HH \triangleq \KK^{-1}= {\rm diag}(1/k_1, 1/k_2, \ldots, 1/k_n)$ is changed to $\HH' =\HH+d\EE_{v,v}$ with $d>0$.
	
	Define $\XX=\mathbf{\Phi}=(\HH\LL+\II)^{-1}$ and $\YY=(\HH'\LL+\II)^{-1}$, and $\Delta \XX=\YY-\XX$. Then, we have
	$\XX(\HH\LL+\II)=(\XX+\Delta\XX)((\HH^{-1}+d\EE_{v,v})\LL+\II)$.
	Expanding the above equation gives
	$-\Delta \XX(\HH'\LL+\II)=d \XX \EE_{v,v}\LL$,
	which means $\Delta \XX=-d\XX\EE_{v,v}\LL\YY$. Considering $\YY=\XX+\Delta \XX$, we further  have
	\begin{equation*}
	\Delta \XX = -d(d\XX\EE_{v,v}\LL+\II)^{-1}\XX\EE_{v,v}\LL\XX.
	\end{equation*}
	
	Define vector $\xx=\XX\ee_v$ and vector $\yy=\LL\ee_v$. According to Sherman-Morrison formula~\cite{Me73}, we have
	\begin{equation*}
	(d\XX\EE_{v,v}\LL+\II)^{-1}=(\II+d\xx\yy^\top)^{-1}=\II-\frac{d\xx\yy^{\top}}{1+d\yy^\top\xx}.
	\end{equation*}
	Then, the $i$-th element of vector $\Delta \XX\ee_v$ at the $v$-th column of matrix $\Delta \XX$ is
	\begin{equation*}
	\ee_{i}^\top \Delta \XX \ee_v =-\frac{d\ee_i^\top{\xx}{\yy}^\top\xx}{1+d{\yy}^\top\xx}.
	\end{equation*}
	Since the elements in $\XX=\mathbf{\Phi}$ are positive,  $\ee_i^\top\xx>0$ always holds. For $\yy^\top\xx$, we have
	\begin{equation}
	\yy^\top\xx=\ee^\top_v\LL\XX\ee_v
	=\sum_{j=1}^n l_{vj}\phi_{jv}>\sum_{j=1}^n l_{v\,j}\phi_{vv}=0.
	\end{equation}
	Thus  the elements in the $v$-th column of $\Delta \XX$ are less than 0, implying that  the elements in the $v$-th column of $\Phi$ will decrease if the stubbornness $k_v$  is reduced.
	
	For the $j$-th column of the matrix $\Phi$, $j\neq v$, using a similar approach one can obtain
	\begin{equation}
	\ee_i^\top \Delta \XX \ee_j = -\frac{d\ee_i^\top\XX\ee_v \ee_v^\top \LL\XX\ee_j}{1+d\yy^\top\xx}.
	\end{equation}
	It is clear that $\yy^\top \xx>0$ and $\ee_i^\top\XX\ee_v>0$. Then we only need to evaluate the elements in $\LL\XX$ that can be recast as
	\begin{equation}
	\LL\XX=\LL(\LL+\KK)^{-1}\KK=\KK-\KK(\LL+\KK)^{-1}\KK.
	\end{equation}
	Since $\KK$ is a diagonal matrix and $(\LL+\KK)^{-1}$ is a positive matrix, all  the non-diagonal elements of matrix $\LL\XX$ are less than 0, which means $\ee_v^\top \LL\XX\ee_j<0$ for any $j\neq v$. Thus, for every  $j\neq v$, all the elements  in the $j$-th column of matrix $\Delta \XX$ are greater than 0, implying that all  the elements  in the $j$-th column of matrix $\Phi$ will increase if we reduce the stubbornness $k_v$.
\end{proof}

Therefore, for opinion dynamics on an undirected graph $\calG= (V,E,w)$,  if we increase the stubbornness $k_v$ of node $v  \in V$, while fixing the stubbornness of all other nodes, all nodes (including node $v$) will place more weight on the initial opinion of node $v$ to form their  equilibrium  expressed opinions.


\section{Metrics for Related Phenomena}\label{measure}

Under the influence of  stubbornness for  individuals, equilibrium  opinions   in the FJ model do not reach consensus, leading to some social phenomena, such as conflict, disagreement, polarization. For the case that nodes have unit stubbornness,
These  phenomena have been quantified and studied in prior works~\cite{DaGoLe13,ChLiDe18,MuMuTs18,XuBaZh21,YiSt20}. However, for the case that  node stubbornness is heterogeneous, related measures and algorithms for these quantities are still less studied.

In this section, we extend previous  quantitative measures for social phenomena to the FJ  opinion model on undirected graphs described in~\eqref{SA02}, which incorporates  heterogeneous stubbornness of nodes. Moreover, we express these quantities in terms of quadratic forms and provide a conservation law governing them.
\subsection{Definitions for Metrics}
In the FJ model, the internal opinions and expressed opinions for individuals are often  different, the extent of which is measured by internal conflict defined below.
\begin{definition}
For opinion dynamics on an undirected graph $\calG= (V,E,w)$ with stubborn matrix $\KK$, its internal conflict $C(\calG)$ is the weighted sum of squares of the differences between internal and expressed opinions over all nodes in $V$, with the weights being  their corresponding   stubbornness:
\begin{equation}\label{eq:dfn_Ci}
C(\calG) = \sum_{i\in V}k_i\left(z_{i}-s_{i}\right)^{2}.
\end{equation}
\end{definition}
Note that~\eqref{eq:dfn_Ci} reduces to the expression for internal conflict in~\cite{ChLiDe18}, when $k_i=1$ for all $i \in V$.

\begin{definition}~\cite{MuMuTs18, DaGoLe13}
For opinion dynamics on an undirected graph $\calG= (V,E,w)$ with stubborn matrix $\KK$, the disagreement between two nodes $i$ and $j$ is defined as $w_{ij} (z_i-z_j)^2$. The disagreement of the whole graph $\calG$ is the sum of the squared differences between all pairs of nodes, given by
\begin{equation}\label{eq:dfn_disagree}
D(\calG) =   \sum\limits_{(i,j) \in E, i<j} w_{ij} (z_i-z_j)^2.
\end{equation}
\end{definition}
Disagreement $D(\calG)$ is also called the external conflict of graph $\calG$~\cite{ChLiDe18}.

In the FJ model, the expressed opinion of a node often deviates from the (weighted) average $\textbf{1}^{\top}\KK\zz=0$ of  the expressed opinions for all nodes. We use polarization to measure how  expressed opinions  deviate from their weighted average in the equilibrium.
\begin{definition}
For opinion dynamics on an undirected graph $\calG= (V,E,w)$ with stubborn matrix $\KK$, the polarization $P(\calG)$ is defined to be:
\begin{equation}\label{eq:dfn_polarization}
P(\calG)= \sum\limits_{i \in V}k_i z_i^2.
\end{equation}
\end{definition}
Equation~\eqref{eq:dfn_polarization} is reduced (3) in~\cite{MuMuTs18}, when $k_i=1$ for all $i \in V$. 





There is a trade-off between disagreement and polarization,   the sum of which is called polarization-disagreement index~\cite{MuMuTs18}.
\begin{definition}
For opinion dynamics on an undirected graph $\calG= (V,E,w)$ with stubborn matrix $\KK$, the  polarization-disagreement index $I_{\rm pd}(\calG)$ is the sum of the  polarization  $P(\calG)$ and disagreement  $D(\calG)$ :
$I_{\rm pd}(\calG) = P(\calG)+ D(\calG)$.
\end{definition}
It is easy to derive that the  polarization-disagreement index $I_{\rm pd}(\calG)$ equals $\sum_{i=1}^n k_i s_{i} z_{i}$, which reduces to the result in~\cite{MuMuTs18} when $\KK=\II$.

\subsection{Expressions in Terms of Quadratic Forms}
The above-defined quantities for social phenomena can be expressed by  quadric forms of  related matrices.
\begin{proposition}\label{Prop}
For opinion dynamics on an undirected graph $\calG= (V,E,w)$ with Laplacian matrix $\LL$, stubborn matrix $\KK$ and internal opinion vector $\vs$, the quantities  $C(\calG)$,  $D(\calG)$,   $P(\calG)$, and  $I_{\rm pd}(\calG)$ can be expressed,  in terms of quadratic forms as:
\begin{equation}\label{eq:dfn_CiA}
C(\calG) = \vs^\top \KK(\LL+\KK)^{-1} \LL\KK^{-1}\LL(\LL+\KK)^{-1} \KK\vs,
\end{equation}
\begin{equation}\label{eq:dfn_disagreeA}
D(\calG) = \vs^{\top}\KK(\LL+\KK)^{-1} \LL(\LL+\KK)^{-1}\KK \vs,
\end{equation}
\begin{equation}\label{eq:dfn_polarizationA}
P(\calG)  =\vs^{\top}\KK(\LL+\KK)^{-1} \KK(\LL+\KK)^{-1}\KK \vs,
\end{equation}
\begin{equation}\label{eq:dfn_DisConA}
I_{\rm pd}(\calG) = \vs^{\top}\KK(\LL+\KK)^{-1} \KK \vs.
\end{equation}
\end{proposition}

The quantities concerned  obey the following conservation law:
\begin{equation}\label{conflictconstant}
    C(\calG)+2D(\calG)+P(\calG)=\sum_{i\in V}k_is_i^2,
\end{equation}
which extends the result in~\cite{MuMuTs18} when $\KK=\II$.

\section{Algorithm for Related Quantities}\label{algorithm}


In this section, we address the problem of fast calculation for  those quantities  defined in last section, including  $C_(\calG)$,   $D(\calG)$,   $P(\calG)$, and  $I_{\rm pd}(\calG)$. Proposition~\ref{Prop} shows that  exactly computing  the quantities  involves inverting  matrix $\LL+\KK$, which takes $O(n^3)$  time and is computationally unacceptable for large graphs. In order to tackle the  computational barrier,  a fast approximation algorithm is developed, which has a nearly linear time complexity with respect to the number of edges in $\calG$. We first  explicitly express  the concerned quantities  as  $\ell_2$ norms of vectors.

\begin{lemma}\label{PropA}
Let $\WW^{1/2}$ be a diagonal matrix defined as $\WW^{1/2}= {\rm diag}(\sqrt{w_1}, \sqrt{w_2}, \ldots, \sqrt{w_m})$. For opinion dynamics on an undirected graph $\calG= (V,E,w)$ with Laplacian matrix $\LL$, stubborn matrix $\KK$ and internal opinion vector $\vs$, the quantities $C(\calG)$,   $D(\calG)$, $P(\calG)$,  and $I_{\rm pd}(\calG)$ can be expressed, respectively, in terms of  $\ell_2$  norms as:
\begin{equation}\label{eq:dfn_CiL}
C(\calG) = \zz^\top \LL \KK^{-1} \LL \zz =\| \KK^{-1/2} \LL (\LL+\KK)^{-1} \vs\| ^2,
\end{equation}
\begin{equation}\label{eq:dfn_disagreeL}
D(\calG) = \zz^{\top} \LL \zz=\| \WW^{1/2} \BB (\LL+\KK)^{-1} \vs\| ^2,
\end{equation}
\begin{equation}\label{eq:dfn_polarizationL}
P(\calG)  ={\zz}^{\top} \KK {\zz}=\| \KK^{1/2} (\LL+\KK)^{-1} {\vs}\| ^2,
\end{equation}
\begin{equation}\label{eq:dfn_DisConL}
I_{\rm pd}(\calG) = \| \WW^{1/2} \BB (\LL+\KK)^{-1} \vs\| ^2+\| \KK^{-1/2} (\LL+\KK)^{-1} \vs\| ^2.
\end{equation}
\end{lemma}

 Since directly calculating the $\ell_2$ norms still requires  inverting  matrix $\LL+\KK$, below we resort to
the linear system solvers~\cite{KySa16} that can efficiently avoid the inverse operation, thus  significantly  reducing the computational complexity.
 \begin{lemma}~\cite{KySa16}\label{ST}
There is a nearly linear time solver $\yy=\textsc{Solve}\left(\TT,\xx,\delta\right)$, which takes   an $n \times n$ positive semi-definite matrix $\TT$ with $m$ non-zero entries, a column vector $\xx$, and an accuracy parameter $\delta$, and returns   a column vector $\yy$ satisfying
$\|\yy-\TT^{\dagger}\xx\|_{\TT}\le\delta\|\TT^{\dagger}\xx\|_{\TT}$, where $\|\vv\|_{\TT}=\sqrt{\vv^\top \TT \vv}$ and $\TT^{\dagger}$ is the pseudo-inverse of matrix $\TT$. The expected time for performing this solver is $O\left(m\log^{3}{n}\log\left(\frac{1}{\delta}\right)\right)$.
\end{lemma}

We now  exploit  Lemma~\ref{ST} to obtain $\eps$-approximations  for  the concerned quantities.

\begin{lemma}\label{lm}
For opinion dynamics on an undirected weighted graph $\calG=(V,E,w)$ with  each edge weight in the interval $[w_{\rm min}, w_{\rm max}]$, Laplacian matrix $\LL$, incident matrix $\BB$, diagonal edge weight matrix $\WW$, parameter $\epsilon \in (0, \frac{1}{2})$,  internal opinion vector $\vs = (s_1, s_2, \ldots, s_n)^{\top}$, and stubbornness matrix $\KK$ obeying relation $\textbf{1}^{\top}\KK\vs=0$, let $\qq = \textsc{Solve}\kh{\LL+\KK, \vs, \delta}$, then the following relations hold: 
	\begin{align}\label{lm1}
	&(1-\epsilon ) \mo{\KK^{1/2}(\LL+\KK)^{-1}\vs}^2 \nonumber\\
	\leq& \mo{\KK^{1/2}\qq}^2 \leq (1+\epsilon ) \mo{\KK^{1/2}(\LL+\KK)^{-1}\vs}^2,
	\end{align}
if $\delta$ satisfies 
$\delta \leq \delta_1 = \frac{\epsilon}{3\sqrt{k_{\min}^{-1}k_{\max}^{-1}(k_{\max} + n w_{\rm max})}};$
	\begin{align}\label{lm2}
	&(1-\epsilon ) \mo{\WW^{1/2}\BB(\LL+\KK)^{-1}{\vs}}^2 \nonumber\\
	\leq & \mo{\WW^{1/2}\BB\qq}^2\leq (1+\epsilon ) \mo{\WW^{1/2}\BB(\LL+\KK)^{-1}{\vs}}^2,
	\end{align}
if $\delta$ satisfies  
  $\delta \leq	\delta_2= \frac{\epsilon k_{\min}\mo{{\vs}} }{3 n(k_{\max}+nw_{\max})} \sqrt{\frac{ w_{\rm min} }{n(k_{\max}+nw_{\max})}}$;
	\begin{align}\label{lm3}
	&\quad(1-\epsilon )\mo{\KK^{-1/2}\LL(\LL+\KK)^{-1}{\vs}}^2 \nonumber\\
	&\leq \mo{\KK^{-1/2}\LL\qq}^2 \leq (1+\epsilon ) \mo{\KK^{-1/2}\LL(\LL+\KK)^{-1}{\vs}}^2,
	\end{align}
if $\delta$ satisfies 
\begin{align*}
  \delta \leq	\delta_3 = \frac{\epsilon w_{\rm min}k_{\min}\sqrt{k_{\min}}\mo{{\vs}}}{3 w_{\rm max} n^3 (k_{\max}+nw_{\max}) \sqrt{nk_{\max}(k_{\max}+nw_{\max})}}.
	\end{align*}
\end{lemma}
\begin{proof}
	By Lemma~\ref{ST}, we have
	\begin{align*}
	\mo{\yy - (\LL+\KK)^{-1}\vs}^2_{\LL+\KK} \leq \delta^2 \mo{(\LL+\KK)^{-1}\vs}_{\LL+\KK}^2.
	\end{align*}	
	The term on the left-hand side (lhs) is bounded as
	\begin{align*}
	&\mo{\qq - (\LL+\KK)^{-1}\vs}_{\LL+\KK}^2
	\geq  k_{\max}^{-1} \mo{\KK^{1/2}\qq - \KK^{1/2}(\LL+\KK)^{-1}\vs}^2 \\
	\geq & k_{\max}^{-1} \len{\mo{\KK^{1/2}\qq} - \mo{\KK^{1/2}(\LL+\KK)^{-1}\vs}}^2,
	\end{align*}
	while the term on the right-hand side (rhs) is bounded by
	\begin{align*}
	\mo{(\LL+\KK)^{-1}\vs}^2_{\LL+\KK} 
	\leq k_{\min}^{-1}(k_{\max} + n w_{\rm max})\mo{\KK^{1/2}(\LL+\KK)^{-1}\vs}^2.
	\end{align*}
	Combining the above  results leads to
	\begin{align*}
	&\len{\mo{\KK^{1/2}\qq} - \mo{\KK^{1/2}(\LL+\KK)^{-1}\vs}}^2 \\
	\leq& \delta^2 k_{\min}^{-1}k_{\max}^{-1}(k_{\max} + n w_{\rm max})\mo{\KK^{1/2}(\LL+\KK)^{-1}\vs}^2,
	\end{align*}
	which implies
	\begin{align*}
	&\frac{\len{\mo{\KK^{1/2}\qq} - \mo{\KK^{1/2}(\LL+\KK)^{-1}\vs}}}{\mo{\KK^{1/2}(\LL+\KK)^{-1}\vs}} \\
	\leq&\sqrt{\delta^2 k_{\min}^{-1}k_{\max}^{-1}(k_{\max} + n w_{\rm max})} \leq \frac{\epsilon}{3}
	\end{align*}
	and
	\begin{align*}
	(1-\frac{\epsilon}{3} )^2 \mo{\KK^{1/2}(\LL+\KK)^{-1}\vs}^2  \leq \mo{\KK^{1/2}\qq}^2 \leq (1+\frac{\epsilon}{3} )^2 \mo{\KK^{1/2}(\LL+\KK)^{-1}\vs}^2.
	\end{align*}
	Using $0 < \epsilon < \frac{1}{2}$, one obtains
	\begin{align*}
	(1-\epsilon ) \mo{\KK^{1/2}(\LL+\KK)^{-1}\vs}^2 \leq \mo{\KK^{1/2}\qq}^2 \leq (1+\epsilon ) \mo{\KK^{1/2}(\LL+\KK)^{-1}\vs}^2,
	\end{align*}
	which finishes the proof.
\end{proof}

The proofs for~\eqref{lm2} and~\eqref{lm3} are similar to that of~\eqref{lm1}, and are thus omitted.  
On the basis of  Lemmas~\ref{ST} and~\ref{lm}, we design a fast  algorithm \textsc{Approxim} to estimate  $C(\calG)$, $D(\calG)$,  $P(\calG)$, and  $I_{\rm pd}(\calG)$ for opinion dynamics on an undirected weighted graph $\calG$. In Algorithm~\ref{alg:VC} we present the pseudocode of \textsc{Approxim}, where  $\delta$ is less than or equal to the minimum of $\delta_1,\delta_2,\delta_3$ defined in  Lemma~\ref{lm}. Theorem~\ref{ThmV} summarizes the performance of algorithm \textsc{Approxim}.

\begin{algorithm}
	\caption{$\textsc{Approxim}\kh{\calG, \vs, \epsilon}$}
	\label{alg:VC}
	\Input{$\calG$: a graph with edge weight in $[ w_{\rm min},w_{\rm max} ]$ 	\\	$\vs$: initial opinion vector with $\bm{1}^{\top}\KK{\vs}=0$ \\$\epsilon$: the error parameter in $ (0, \frac{1}{2})$ \\
	}
	\Output{$\{  \tilde{D}(\calG),  \tilde{P}(\calG),  \tilde{C}(\calG), \tilde{I}_{\rm pd}(\calG)\}$
	}
	$\delta = \frac{\epsilon w_{\rm min}k_{\min}\sqrt{k_{\min}}\mo{\vs}}{3 w_{\rm max} n^3 (k_{\max}+nw_{\max}) \sqrt{nk_{\max}(k_{\max}+nw_{\max})}}$ \;
	$ \qq  = \textsc{Solve}(\LL+\KK, \vs, \delta)$ \;
	$ \tilde{C}(\calG) = \mo{\KK^{-1/2}\LL \qq }^2$ \;
	$\tilde{D}(\calG) = \mo{\WW^{1/2}\BB \qq }^2$ \;
    $\tilde{P}(\calG)= \mo{\KK^{1/2}\qq }^2$ \;
	$ \tilde{I}_{\rm pd}(\calG) = \tilde{P}(\calG) + \tilde{D}(\calG)$ \;
	\textbf{return} $\{ \tilde{C}(\calG), \tilde{D}(\calG),  \tilde{P}(\calG), \tilde{I}_{\rm pd}(\calG) \}$ \;
\end{algorithm}

\begin{theorem}\label{ThmV}
For opinion dynamics on an undirected weighted graph $\calG=(V,E,w)$ having  $n$ nodes,  $m$ edges with  weight of each edge being in the interval $[w_{\rm min}, w_{\rm max}]$,  an error   parameter $\epsilon \in (0, \frac{1}{2})$, and the internal opinion vector $\vs$,  the algorithm $\textsc{Approxim}\kh{\calG, \vs, \epsilon}$ runs in expected time  $O\left(m\log^{4}{n}\log\left(\frac{r}{\epsilon}\right)\right)$ where $r=\frac{w_{\rm max}}{w_{\rm min}}$. $\textsc{Approxim}$ returns  $\tilde{C}(\calG)$, $\tilde{D}(\calG)$,  $\tilde{P}(\calG)$,  $\tilde{I}_{\rm pd}(\calG)$ for the the $\eps$-approximation of  internal conflict $C(\calG)$,  disagreement  $D(\calG)$,  polarization $P(\calG)$ and polarization-disagreement index $I_{\rm pd}(\calG)$, satisfying
$\tilde{C}(\calG) \approx_\eps{C(\calG)}$,
$\tilde{D}(\calG) \approx_\eps{D(\calG)}$,
$\tilde{P}(\calG) \approx_\eps{P(\calG)}$, and
$\tilde{I}_{\rm pd}(\calG) \approx_\eps I_{\rm pd}(\calG)$.
\end{theorem}

\begin{table*}[]
\fontsize{7}{7.4}\selectfont
\caption{Statistics of real datasets and comparison of running time (seconds, $s$) between \textsc{Exact} and \textsc{Approxim}.}\label{runningtime}  
\begin{tabular}{@{}ccccccccccc@{}}
\toprule
 &  &  & \multicolumn{8}{c}{Running time ($s$) for algorithms \textsc{Exact} and \textsc{Approxim}} \\ \cmidrule(l){4-11}
\multirow{2}{*}{Network} & \multirow{2}{*}{$n$} & \multirow{2}{*}{$m$} & \multicolumn{2}{c}{Uniform} & \multicolumn{2}{c}{Power-law} & \multicolumn{2}{c}{Normal} & \multicolumn{2}{c}{Exponential} \\ \cmidrule(l){4-5} \cmidrule(l){6-7} \cmidrule(l){8-9} \cmidrule(l){10-11}
 &  &  & \textsc{Exact} & \textsc{Approxim} & \textsc{Exact} & \textsc{Approxim} & \textsc{Exact} & \textsc{Approxim} & \textsc{Exact} & \textsc{Approxim} \\ \midrule
GrQc & 4158 & 13422 & 1.44 & 2.42 & 1.41 & 2.40 & 1.45 & 2.40 & 1.43 & 2.40 \\
USgrid & 4941 & 6594 & 2.24 & 2.41 & 2.22 & 2.36 & 2.25 & 2.35 & 2.23 & 2.35 \\
Erdos992 & 5094 & 7515 & 2.37 & 2.42 & 2.37 & 2.40 & 2.46 & 2.40 & 2.51 & 2.41 \\
Advogato & 5167 & 39432 & 2.51 & 2.44 & 2.55 & 2.42 & 2.52 & 2.42 & 2.56 & 2.43 \\
Bcspwr10 & 5300 & 8271 & 2.71 & 2.37 & 2.73 & 2.60 & 2.72 & 2.41 & 2.69 & 2.40 \\
Reality & 6809 & 7680 & 6.01 & 2.37 & 6.55 & 2.36 & 5.98 & 2.38 & 5.95 & 2.36 \\
PagesGovernment & 7057 & 89429 & 6.75 & 2.42 & 6.84 & 2.42 & 7.07 & 2.41 & 8.04 & 2.34 \\
WikiElec & 7115 & 100753 & 7.45 & 2.43 & 6.93 & 2.43 & 8.01 & 2.42 & 7.18 & 2.43 \\
Dmela & 7393 & 25569 & 8.87 & 2.42 & 8.51 & 2.42 & 9.23 & 2.41 & 8.65 & 2.41 \\
HepPh & 11204 & 117619 & 30.21 & 2.37 & 28.49 & 2.35 & 29.15 & 2.36 & 28.25 & 2.35 \\
Anybeat & 12645 & 49132 & 40.94 & 2.31 & 41.39 & 2.31 & 41.74 & 2.30 & 41.49 & 2.34 \\
PagesCompany & 14113 & 52126 & 56.67 & 2.32 & 57.07 & 2.31 & 59.83 & 2.43 & 59.20 & 2.33 \\
AstroPh & 17903 & 196972 & 122.64 & 2.44 & 119.47 & 2.44 & 114.35 & 2.44 & 118.19 & 2.43 \\
CondMat & 21363 & 91286 & 205.04 & 2.36 & 202.79 & 2.35 & 203.85 & 2.37 & 202.91 & 2.35 \\
Gplus & 23628 & 39194 & 276.76 & 2.63 & 275.54 & 2.63 & 279.40 & 2.79 & 280.99 & 2.69 \\
GemsecRO & 41773 & 125826 & 1584.76 & 3.05 & 1573.48 & 3.29 & 1577.72 & 3.05 & 1599.97 & 3.09 \\
GemsecHU & 47538 & 222887 & 2448.57 & 4.59 & 2537.67 & 4.65 & 2561.06 & 4.43 & 2480.51 & 4.44 \\
WikiTalk & 92117 & 360767 & - & 3.06 & - & 3.16 & - & 3.11 & - & 3.13 \\
Buzznet & 101163 & 2763066 & - & 5.70 & - & 6.01 & - & 5.99 & - & 6.02 \\
LiveMocha & 104103 & 2193083 & - & 6.33 & - & 6.09 & - & 6.09 & - & 6.04 \\
Douban & 154908 & 327162 & - & 3.18 & - & 3.02 & - & 3.15 & - & 3.17 \\
Gowalla & 196591 & 950327 & - & 4.07 & - & 4.02 & - & 4.07 & - & 4.06 \\
Academia & 200169 & 1022441 & - & 4.27 & - & 4.34 & - & 4.35 & - & 4.44 \\
GooglePlus & 211186 & 1141650 & - & 4.17 & - & 4.08 & - & 4.07 & - & 4.37 \\
Pwtk & 217891 & 5653221 & - & 8.16 & - & 8.33 & - & 8.24 & - & 8.06 \\
Citeseer & 227320 & 814134 & - & 3.74 & - & 3.61 & - & 3.66 & - & 3.71 \\
MathSciNet & 332689 & 820644 & - & 4.32 & - & 4.24 & - & 4.14 & - & 4.36 \\
TwitterFollows & 404719 & 713319 & - & 3.64 & - & 3.67 & - & 3.67 & - & 3.63 \\
Flickr & 513969 & 3190452 & - & 7.42 & - & 7.99 & - & 7.68 & - & 7.84 \\
Delicious & 536108 & 1365961 & - & 4.92 & - & 5.09 & - & 5.06 & - & 5.00 \\
FourSquare & 639014 & 3214986 & - & 6.78 & - & 6.70 & - & 6.78 & - & 6.63 \\
Digg & 770799 & 5907132 & - & 12.59 & - & 12.61 & - & 12.51 & - & 12.25 \\
IMDB & 896305 & 3782447 & - & 13.21 & - & 13.99 & - & 13.17 & - & 13.59 \\
Ldoor & 909537 & 20770807 & - & 23.37 & - & 23.31 & - & 23.26 & - & 23.24 \\
RoadNetPA & 1087562 & 1541514 & - & 6.65 & - & 7.00 & - & 6.47 & - & 6.68 \\
YoutubeSnap & 1134890 & 2987624 & - & 9.11 & - & 8.56 & - & 9.85 & - & 8.95 \\
Lastfm & 1191805 & 4519330 & - & 11.81 & - & 11.50 & - & 11.16 & - & 11.58 \\
Pokec & 1632803 & 22301964 & - & 95.00 & - & 96.11 & - & 94.52 & - & 95.13 \\
RoadNetCA & 1957027 & 2760388 & - & 10.26 & - & 10.88 & - & 10.57 & - & 10.17 \\
Flixster & 2523386 & 7918801 & - & 19.99 & - & 20.27 & - & 20.75 & - & 19.96 \\
Patent & 3774768 & 16518947 & - & 80.63 & - & 78.45 & - & 77.77 & - & 80.85 \\
LiveJournal & 4033137 & 27933062 & - & 108.10 & - & 97.97 & - & 101.74 & - & 104.98 \\
ItalyOsm & 6686493 & 7013978 & - & 22.18 & - & 23.62 & - & 23.08 & - & 22.28 \\ \bottomrule
\end{tabular}
\end{table*}
\section{Experiments}

In this section, we evaluate the efficiency and effectivity of our approximation algorithm \textsc{Approxim}. To achieve this goal, we implement this algorithm on a large set of real networks, and  compare the accuracy and  running time of \textsc{Approxim} with  an exact algorithm \textsc{Exact}, which  computes relevant quantities via directly inverting the matrix $\LL+\KK$, performing matrix product of related matrices, and  calculating related  $\ell_2$ norms.

\textbf{Machine and Repeatability.} All our  experiments  were  run on a Linux box  with 4-core 4.2GHz  Intel i7-7700K CPU and   32GB of main memory. In  our experiments,  the error parameter $\epsilon$ is set to be $10^{-6}$.  Both approximation algorithm \textsc{Approxim} and  exact algorithm \textsc{Exact} are programmed with \textit{Julia v1.5.1} using a single thread. The  solver \textsc{Solve}  used in this paper is based on the method in~\cite{KySa16},  whose \textit{Julia language} implementation is open and accessible  on the website \url{https://github.com/danspielman/Laplacians.jl}. Our code is available at \url{https://anonymous.4open.science/r/heter_stubbornness}.

\textbf{Datasets.} The real network datasets used in our experiments are  publicly available in the Koblenz Network Collection~\cite{Ku13}  and Network Repository~\cite{RyNe15}. The statistics for basic information of these  datasets are listed  in the first three columns of Table~\ref{runningtime}. 

\textbf{Distributions of Innate opinions.} We use  four  different distributions of the innate opinions in our experiments, including: uniform distribution, power-law distribution, normal distribution, and exponential distribution. The innate opinion $s_i$ for each node $i$ is in interval $[-1,1]$. 

\textbf{Efficiency and Scalability.} In Table~\ref{runningtime}, we report the running time of algorithms \textsc{Approxim} and \textsc{Exact} on different networks  in order to compare their computation complexity. Note that for  large networks  we cannot compute related quantities using algorithm  \textsc{Exact}, due to the high cost of memory and time, but we can perform  algorithm  \textsc{Approxim}.  For each distribution of the innate opinions, we record the running times of \textsc{Approxim} and \textsc{Exact} in different networks.  Table~\ref{runningtime} shows that for all networks considered, the running time of \textsc{Approxim} is smaller than that of \textsc{Exact}.  For moderately large networks with more than ten thousand nodes,  \textsc{Approxim} is orders of magnitude faster than \textsc{Exact}.  In particular, \textsc{Approxim} is scalable to  large network swith more than one million nodes. 

\textbf{Accuracy.} 
In Table~\ref{relativeerror}, we present the accuracy results for  algorithm \textsc{Approxim}. For each of the four distributions of innate opinions, we compare the  results of \textsc{Approxim} with those of \textsc{Exact} for some networks in Table~\ref{runningtime}. For each relevant quantity $\rho$, we apply the mean relative error  $\sigma = | \rho - \tilde{\rho} | /\rho$ of   $\tilde{\rho}$ obtained by \textsc{Approxim} as an estimation of  $\rho$. Table~\ref{relativeerror} reports   the mean relative errors of the four estimated quantities,  $\tilde{C}(\calG)$,  $\tilde{D}(\calG)$,   $\tilde{P}(\calG)$,  and $\tilde{I}_{pd}(\calG)$,  
for many real networks. We can see that for all quantities concerned,  the  results of \textsc{Approxim} are very close to those associated with  \textsc{Exact}, since their actual relative errors are negligible,  which are all  less than $10^{-8}$.

\begin{table*}[!hb]
\fontsize{6}{8.5}\selectfont
\caption{Relative error for estimated $C(\calG)$, $D(\calG)$, $P(\calG)$, $I_{\rm pd}(\calG)$ for four internal distributions with input parameter $\epsilon=10^{-6}$.}\label{relativeerror}
\setlength\tabcolsep{3pt}
\begin{tabular}{@{}lllllllllllllllll@{}}
\toprule
\multicolumn{1}{c}{\multirow{3}{*}{Network}} & \multicolumn{16}{c}{\begin{tabular}[c]{@{}c@{}} Relative error of four estimated quantities for four internal opinion distributions \end{tabular}} \\ \cmidrule(l){2-17} 
\multicolumn{1}{c}{} & \multicolumn{4}{c}{Uniform Distribution} & \multicolumn{4}{c}{Power-law distribution} & \multicolumn{4}{c}{Normal distribution} & \multicolumn{4}{c}{Exponential distribution} \\ \cmidrule(l){2-5} \cmidrule(l){6-9} \cmidrule(l){10-13} \cmidrule(l){14-17} 
\multicolumn{1}{c}{} & \multicolumn{1}{c}{$C(\calG)$} & \multicolumn{1}{c}{$D(\calG)$} & \multicolumn{1}{c}{$P(\calG)$} & \multicolumn{1}{c}{$I_{\rm pd}(\calG)$} & \multicolumn{1}{c}{$C(\calG)$} & \multicolumn{1}{c}{$D(\calG)$} & \multicolumn{1}{c}{$P(\calG)$} & \multicolumn{1}{c}{$I_{\rm pd}(\calG)$} & \multicolumn{1}{c}{$C(\calG)$} & \multicolumn{1}{c}{$D(\calG)$} & \multicolumn{1}{c}{$P(\calG)$} & \multicolumn{1}{c}{$I_{\rm pd}(\calG)$} & \multicolumn{1}{c}{$C(\calG)$} & \multicolumn{1}{c}{$D(\calG)$} & \multicolumn{1}{c}{$P(\calG)$} & \multicolumn{1}{c}{$I_{\rm pd}(\calG)$} \\ \midrule
GrQc & 2.69E-09 & 1.79E-10 & 1.23E-10 & 1.04E-14 & 1.94E-09 & 6.12E-10 & 5.42E-10 & 8.51E-16 & 6.00E-09 & 2.39E-09 & 1.30E-09 & 3.03E-15 & 2.42E-08 & 4.51E-11 & 1.84E-11 & 2.32E-14 \\
USgrid & 3.61E-08 & 1.05E-08 & 4.56E-09 & 2.89E-14 & 7.72E-10 & 3.02E-09 & 4.90E-09 & 3.78E-15 & 3.48E-08 & 1.27E-08 & 5.74E-09 & 2.71E-14 & 3.19E-09 & 4.43E-09 & 3.53E-09 & 1.24E-14 \\
Erdos992 & 2.45E-08 & 8.04E-10 & 4.30E-10 & 4.71E-15 & 2.23E-10 & 3.12E-10 & 5.09E-10 & 4.94E-16 & 5.01E-09 & 9.26E-10 & 5.87E-10 & 7.81E-15 & 9.41E-09 & 7.38E-10 & 4.17E-10 & 1.66E-15 \\
Advogato & 1.78E-08 & 6.70E-09 & 3.21E-09 & 4.58E-15 & 4.66E-09 & 2.03E-09 & 1.70E-09 & 2.76E-16 & 2.09E-09 & 2.01E-10 & 1.60E-10 & 2.02E-15 & 9.51E-10 & 5.41E-10 & 4.59E-10 & 1.91E-16 \\
Bcspwr10 & 2.56E-08 & 1.43E-08 & 6.28E-09 & 2.01E-14 & 8.02E-09 & 3.19E-09 & 1.58E-09 & 4.45E-14 & 3.04E-09 & 7.10E-10 & 4.17E-10 & 1.19E-14 & 1.07E-08 & 2.02E-09 & 1.18E-09 & 7.24E-15 \\
Reality & 7.43E-10 & 6.97E-10 & 3.32E-10 & 4.26E-16 & 9.62E-11 & 8.92E-10 & 1.05E-09 & 1.70E-15 & 4.95E-09 & 2.31E-11 & 1.24E-11 & 1.69E-16 & 1.75E-09 & 5.53E-10 & 2.95E-10 & 3.49E-16 \\
PagesGovernment & 3.01E-07 & 3.45E-08 & 7.91E-10 & 5.79E-15 & 1.66E-08 & 6.34E-10 & 1.28E-09 & 4.03E-15 & 2.60E-08 & 1.01E-08 & 1.07E-08 & 1.14E-14 & 4.58E-09 & 2.21E-09 & 1.72E-09 & 1.51E-15 \\
WikiElec & 5.37E-11 & 4.91E-11 & 2.48E-11 & 3.92E-16 & 4.98E-09 & 2.46E-10 & 4.58E-10 & 3.87E-16 & 2.67E-08 & 4.58E-09 & 2.59E-09 & 6.57E-16 & 9.24E-09 & 5.78E-11 & 6.80E-11 & 1.81E-15 \\
Dmela & 4.56E-11 & 2.62E-10 & 1.70E-10 & 9.74E-16 & 3.67E-09 & 6.35E-10 & 1.07E-09 & 6.87E-16 & 5.88E-09 & 3.76E-09 & 2.36E-09 & 3.61E-15 & 6.75E-09 & 2.99E-09 & 1.65E-09 & 9.77E-16 \\
HepPh & 1.25E-09 & 6.00E-09 & 3.61E-09 & 1.25E-14 & 4.49E-09 & 8.54E-10 & 2.18E-09 & 1.46E-15 & 1.36E-08 & 6.47E-09 & 5.06E-09 & 8.41E-15 & 7.73E-08 & 4.77E-08 & 9.19E-09 & 1.82E-14 \\
Anybeat & 1.38E-08 & 1.46E-09 & 9.03E-10 & 1.25E-16 & 1.02E-09 & 3.85E-11 & 3.20E-10 & 1.44E-15 & 2.83E-08 & 3.37E-09 & 1.91E-09 & 1.39E-14 & 2.50E-08 & 1.40E-09 & 8.55E-10 & 2.63E-15 \\
PagesCompany & 5.93E-10 & 1.91E-09 & 1.70E-09 & 5.05E-15 & 2.85E-08 & 7.04E-09 & 8.21E-09 & 4.37E-15 & 4.12E-09 & 2.57E-10 & 2.92E-10 & 6.40E-15 & 8.06E-09 & 6.91E-09 & 4.65E-09 & 1.66E-14 \\
AstroPh & 3.44E-09 & 1.08E-09 & 1.36E-09 & 4.58E-16 & 2.68E-08 & 3.50E-08 & 3.99E-09 & 7.67E-15 & 5.04E-09 & 3.18E-09 & 2.36E-09 & 8.49E-16 & 6.12E-09 & 8.24E-10 & 1.08E-09 & 1.89E-14 \\
CondMat & 1.58E-09 & 2.00E-09 & 1.67E-09 & 8.07E-15 & 4.60E-09 & 2.89E-09 & 4.69E-09 & 2.18E-15 & 8.52E-09 & 1.93E-09 & 2.17E-09 & 3.60E-15 & 8.17E-09 & 3.16E-10 & 2.34E-10 & 1.56E-14 \\
Gplus & 2.32E-09 & 3.88E-10 & 2.51E-10 & 6.70E-16 & 3.33E-10 & 1.88E-11 & 5.01E-11 & 7.62E-16 & 6.17E-10 & 6.78E-11 & 5.00E-11 & 3.23E-16 & 1.28E-09 & 1.30E-10 & 1.11E-10 & 2.26E-16 \\
GemsecRO & 9.75E-09 & 4.75E-09 & 4.12E-09 & 2.53E-15 & 4.50E-09 & 1.25E-09 & 2.05E-09 & 7.98E-16 & 4.30E-09 & 3.52E-09 & 3.05E-09 & 2.93E-15 & 4.45E-08 & 4.03E-08 & 1.18E-08 & 7.01E-14 \\
GemsecHU & 9.90E-09 & 2.03E-09 & 3.10E-09 & 1.19E-14 & 2.20E-09 & 1.40E-09 & 1.32E-08 & 2.29E-14 & 5.15E-09 & 2.98E-09 & 3.33E-09 & 1.45E-14 & 1.08E-08 & 4.40E-09 & 5.48E-09 & 3.23E-15\\ \bottomrule
\end{tabular}
\end{table*}


\section{Related Work}\label{related}
In this section, we  briefly review some prior work that lies close to ours.

\textbf{Stubbornness.} As an important individual attribute, stubbornness is a fundamental element affecting diverse aspects of opinion dynamics. For an individual, stubbornness indicates the degree to which the individual sticks to its own internal opinion or willingness to conform to its neighbors' opinions. In fact, stubbornness is equivalent to people's susceptibility to persuasion introduced in~\cite{AbKlPaTs18, AbChKlLiPaSoTs21}. There is a rich body of work~\cite{AlPoMc06,DiKlLe92,EvOaSc92,StMo07, WaPoAr93} in politics and social psychology, studying empirically this key factor, which show that stubbornness can be applied to a large variety of practical scenarios, such as  product marketing, public health campaigns, and political candidates. Some recent work also studied how to change an individual's attitude or stubbornness~\cite{Fo02,IjDeMi06,KaMade09}, and showed that heterogeneity of individual's stubbornness is ubiquitous. It is thus of practical and theoretical interest to analyze the influences of heterogeneous stubbornness of agents on different aspects of opinion dynamics.


\textbf{Opinion dynamics model.} To better understand the formation of opinion dynamics, several relevant models have been presented, among which, FJ model is a popular one. It is frequently adopted to study various aspects of opinion dynamics, and thus has been extensively studied. A sufficient condition for stability of the FJ model was obtained in~\cite{RaFrTeIs15}, and the equilibrium expressed opinion was derived in~\cite{DaGoPaSa13,BiKlOr15}. In the FJ model, the total expressed opinion is equal to the total internal opinion~\cite{GiTeTs13,XuBaZh21}. Some interpretations of the FJ model were provided in~\cite{BiKlOr15} and~\cite{GhSr14}. Most of the previous work only considered homogeneous stubbornness. In~\cite{AbKlPaTs18,AbChKlLiPaSoTs21}, the impact of stubbornness on optimizing the overall opinion dynamics were analyzed by introducing a resistance parameter to modify the FJ model. However, it is still not well understood how heterogeneous stubbornness affects other respects of opinion dynamics, such as convergence speed and overall opinion. Furthermore, some prior interpretations for expressed opinion are not applicable to the case with inhomogeneous stubbornness, for example, electrical networks~\cite{GhSr14}. Our interpretation for spanning diverging forest is novel, which  differs greatly from previous ones.

\textbf{Quantification of social phenomena.} In recent years, social media and online social networks have experienced explosive growth, which is a accompanied by diverse social phenomena, including polarization, disagreement, and conflict. Actually, these phenomena have taken place in human societies for millenia, and have been a recent hot subject of study in different disciplines, especially social science. Thus far, various measures have been developed to quantify these phenomena, such as disagreement~\cite{MuMuTs18,DaGoLe13}, polarization~\cite{ DaGoLe13, MaTeTs17, MuMuTs18}, conflict~\cite{ChLiDe18}, and controversy~\cite{ChLiDe18}. Most of these measures are based on the FJ opinion dynamics model~\cite{FrJo90} with uniform stubbornness, which do not apply to the case in the presence of different stubbornness. To make up for the deficiency, we extend these quantities by incorporating heterogeneous stubbornness, so that they can measure the corresponding social phenomena. The issue of how to quantify these phenomena has received increasing amount of attention. Since direct computation of these indicators involves the operations of matrix inversion and multiplication, making it computationally infeasible for large-scale graphs with millions of nodes, we present a nearly linear time algorithm to estimate all these quantities.

\textbf{Impacts of reciprocity on dynamics.} Note that the FJ model on an undirected network with non-zero heterogeneous stubbornness is equivalent to traditional FJ model on a corresponding weighted directed graph with uniform stubbornness and asymmetric weights between edges linking a pair of nodes. The tendency of node pairs to form mutual asymmetric connection in directed networks is also called  link reciprocity, which is a common characteristic of many realistic networks~\cite{GaLo04,AkVaFa12,WaLiHaStToCh13,SqPiRuGa13}, such as the World Wide Web~\cite{AlJeBa99}, e-mail networks~\cite{EbMiBo02,NeFoBa02}, and world trade web~\cite{SeBo03}. It has been shown the ubiquitous link reciprocity strongly affects dynamical processes running in binary networks, for example, spread of computer viruses~\cite{NeFoBa02} or information~\cite{ZhZhSuTaZhZh14},  percolation~\cite{BoSa05}, and random walks~\cite{ZhLiSh14}. By contrast, the influence of reciprocity on opinion dynamics has attracted much less attention, although it is suggested that reciprocity could play a crucial role in various aspects of this dynamics.

\section{Conclusion}\label{conclusion}

Stubbornness of individuals is a key factor affecting various aspects opinion dynamics. In this paper, we provided a solid theoretical analysis of a classic model for opinion dynamics on undirected graphs in the presence of heterogeneous stubbornness for different nodes. We demonstrated that heterogeneous stubbornness strongly affects almost all aspects of opinion dynamics, such as the expressed opinion, convergence velocity to equilibrium, and overall opinion. For example, when a node increase its stubbornness,  the expressed opinions of nodes place more weight on the internal opinions of this node, showing that a more stubborn individual has more steady social power. Again for instance, when the stubbornness are different among nodes, the sum of expressed opinions is often not equal to the total initial opinion, indicating the difference between the evolution of social power with heterogeneous and homogeneous stubbornness. We provided an interpretation of the expressed opinion in terms of spanning diverging forests of a stubbornness-dependent digraph associated with the original undirected graph.

In addition, we extended previously proposed quantitative indicators for social phenomena in the opinion dynamics model with identical  stubbornness to the case when stubbornness are heterogeneous. The considered social concepts include conflict, disagreement, and polarization, all of which incorporate heterogeneous stubbornness.  Since direct computation of these quantities involves matrix inverse, which is computationally challenging for large graphs. We devised an approximation algorithm to estimate these quantities concerned. Our algorithm has a theoretical guarantee for accuracy and an almost linear computation time complexity with respect to the number of edges. Finally, we executed experiments on many real-world networks, demonstrating the high efficiency and good effectiveness of our algorithm, which is applicable to large graphs with millions of nodes.

\bibliographystyle{ACM-Reference-Format}
\bibliography{stubborn}


\end{document}